\documentclass{sig-alt-release2}

\newcommand{\includefigs}[1]{#1}            
\newcommand{\includefigsAppendix}[1]{}      

\usepackage{color}
\usepackage{multirow}

\newtheorem{thm}{Theorem}[section]
\newtheorem{theorem}[thm]{Theorem}

\newtheorem{lemma}[thm]{Lemma}
\newtheorem{corollary}[thm]{Corollary}

\newcommand{\remove}[1]{}
\newcommand{\lt}{\left}
\newcommand{\rt}{\right}

\CopyrightYear{2012}
\title{Strong Scaling of Matrix Multiplication Algorithms and Memory-Independent Communication Lower Bounds}
\numberofauthors{5}
\author{
\alignauthor
Grey Ballard\titlenote{Research supported by Microsoft (Award $\#$024263) and Intel (Award $\#$024894) funding and by matching funding by U.C. Discovery (Award $\#$DIG07-10227). Additional support comes from Par Lab affiliates National Instruments, Nokia, NVIDIA, Oracle, and Samsung.}\\
\affaddr{UC Berkeley}\\
\email{ballard@eecs.berkeley.edu}
\alignauthor
James Demmel$^{^{^{\textrm{\normalsize $*$}}}}$\titlenote{Research is also supported by DOE grants DE-{}SC0003959, DE-{}SC0004938, and DE-{}AC02-05CH11231.}\\
\affaddr{UC Berkeley}\\
\email{demmel@cs.berkeley.edu}
\alignauthor
Olga Holtz\titlenote{Research supported by the Sofja Kovalevskaja
programme of Alexander von Humboldt Foundation
and by the National Science Foundation under agreement DMS-0635607,
while visiting the Institute for Advanced Study.}\\
\affaddr{UC Berkeley and TU Berlin}\\
\email{holtz@math.berkeley.edu}
\and
Benjamin Lipshitz$^{^{^{\textrm{\normalsize $*$}}}}$\\
\affaddr{UC Berkeley}\\
\email{lipshitz@berkeley.edu}
\alignauthor
Oded Schwartz\titlenote{
Research supported by U.S. Department of Energy grants under
Grant Numbers DE-{}SC0003959.
}\\
\affaddr{UC Berkeley}\\
\email{odedsc@eecs.berkeley.edu}
}
\begin{document}\sloppy
\maketitle

\begin{abstract}
A parallel algorithm has perfect strong scaling if its running time on $P$ processors is linear in $1/P$, including all communication costs. Distributed-memory parallel algorithms for matrix multiplication with perfect strong scaling have only recently been found. One is based on classical matrix multiplication (Solomonik and Demmel, 2011), and one is based on Strassen's fast matrix multiplication (Ballard, Demmel, Holtz, Lipshitz, and Schwartz, 2012). Both algorithms scale perfectly, but only up to some number of processors where the inter-processor communication no longer scales.

We obtain a memory-independent communication cost lower bound on classical and Strassen-based distributed-memory matrix multiplication algorithms.  These bounds imply that no classical or Strassen-based parallel matrix multiplication algorithm can strongly scale perfectly beyond the ranges already attained by the two parallel algorithms mentioned above. The memory-independent bounds and the strong scaling bounds generalize to other algorithms.
\end{abstract}

\thispagestyle{empty}
\bigskip

{{\bf ACM Classification Keywords}}:
F.2.1

{{\bf ACM General Terms}}: Algorithms, Design, Performance.


\clearpage

\section{Introduction}

In evaluating the recently proposed parallel algorithm based on Strassen's matrix multiplication \cite{BallardDemmelHoltzLipshitzSchwartz12a} and comparing the communication costs to the known lower bounds \cite{BallardDemmelHoltzSchwartz11b}, we found a gap between the upper and lower bounds for certain problem sizes.  The main motivation of this work is to close this gap by tightening the lower bound for this case, proving that the algorithm is optimal in all cases, up to $O(\log P)$ factors.  A similar scenario exists in the case of classical matrix multiplication; in this work we provide the analogous tightening of the existing lower bound \cite{IronyToledoTiskin04} to show optimality of another recently proposed algorithm \cite{SolomonikDemmel11}.

In addition to proving optimality of algorithms, the lower bounds in this paper yield another interesting conclusion regarding strong scaling.  We say that an algorithm strongly scales perfectly if it attains running time on \(P\) processors which is linear in \(1/P\), including all communication costs.   While it is possible for classical and Strassen-based matrix multiplication algorithms to strongly scale perfectly, the communication costs restrict the strong scaling ranges much more than do the computation costs.  These ranges depend on the problem size relative to the local memory size, and on the computational complexity of the algorithm.

Interestingly, in both cases the dominance of a memory-independent bound arises, and the strong scaling range ends, exactly when the memory-dependent latency lower bound becomes constant.  This observation may provide a hint as to where to look for strong scaling ranges in other algorithms.  Of course, since the latency cost cannot possibly drop below a constant, it is an immediate result of the memory-dependent bounds that the latency cost cannot continue to strongly scale perfectly.  However the bandwidth cost typically dominates the cost, and it is the memory-independent bandwidth scaling bounds that limit the strong scaling of matrix multiplication in practice.  For simplicity we omit discussions of latency cost, since the number of messages is always a factor of \(M\) below the bandwidth cost in the strong scaling range, and is always constant outside the strong scaling range.

While the main arguments in this work focus on matrix multiplication%
, we present results in such a way that they can be generalized to other algorithms, including other $O(n^3)$-based dense and sparse algorithms as in \cite{BallardDemmelHoltzSchwartz11a} and other fast matrix multiplication algorithms as in \cite{BallardDemmelHoltzSchwartz11b}.

Our paper is organized as follows.  In Section~\ref{sec:Strassen} we prove a memory-independent communication lower bound for Strassen-based matrix multiplication algorithms, and we prove an analogous bound for classical matrix multiplication in Section~\ref{sec:classical}.  We discuss the implications of these bounds on strong scaling in Section~\ref{sec:strongscaling} and compare the communication costs of Strassen and classical matrix multiplication as the number of processors increases.  In Section~\ref{sec:discussion} we discuss generalization of our bounds to other algorithms.  The main results of this paper are summarized in Table~\ref{tbl:summary}.

\section{Communication Lower Bounds}
\label{sec:LB}
We use the distributed-memory communication model (see, \emph{e.g.}, \cite{BallardDemmelHoltzSchwartz11a}), where the bandwidth-cost of an algorithm is proportional to the number of words communicated and the latency-cost is proportional to the number of messages communicated along the critical path.  We will use the notation that \(n\) is the size of the matrices, \(P\) is the number of processors, \(M\) is the local memory size of each processor, and \(\omega_0=\log_2 7\approx 2.81\) is the exponent of Strassen's matrix multiplication.

\subsection{Strassen's Matrix Multiplication}
\label{sec:Strassen}

In this section, we prove a memory-independent lower bound for Strassen's matrix multiplication of $\Omega(n^2/P^{2/\omega_0})$ words, where $\omega_0=\log_2 7$.  We reuse notation and proof techniques from \cite{BallardDemmelHoltzSchwartz11a}.  By prohibiting redundant computations we mean that each arithmetic operation is computed by exactly one processor.  This is necessary for interpreting edge expansion as communication cost. 

\begin{theorem}
\label{thm:Strassen}
Suppose a parallel algorithm performing Strassen's matrix multiplication minimizes computational costs in an asymptotic sense and performs no redundant computation. Then, for sufficiently large $P$,\footnote{The theorem applies to any \(P\geq 2\) with a strict enough assumption on the load balance among vertices in $Dec_{\lg n} C$ as defined in the proof.} some processor must send or receive at least $\Omega\lt(\frac{n^2}{P^{2/w_0}}\rt)$ words.
\end{theorem}
\begin{proof}
The computation DAG (see e.g., \cite{BallardDemmelHoltzSchwartz11a} for formal definition) of Strassen's algorithm multiplying square matrices $A\cdot B=C$ can be partitioned into three subgraphs: an encoding of the elements of $A$, an encoding of the elements of $B$, and a decoding of the scalar multiplication results to compute the elements of $C$.  These three subgraphs are connected by edges that correspond to scalar multiplications.  Call the third subgraph $Dec_{\lg n} C$, where $\lg n=\log_2 n$ is the number of levels of recursion for matrices of dimension $n$.

In order to minimize computational costs asymptotically, the running time for Strassen's matrix multiplication must be $O(n^{\omega_0}/P)$. Since a constant fraction of the flops correspond to vertices in $Dec_{\lg n}C$, this is possible only if some processor performs $\Theta\left(\frac{n^{\omega_0}}{P} \right)$ flops corresponding to vertices in $Dec_{\lg n}C$.

By Lemma 10 of \cite{BallardDemmelHoltzSchwartz11b}, the edge expansion of $Dec_kC$ is given by $h(Dec_kC) = \Omega( (4/7)^k )$.
Using Claim 5 there (decomposition into edge disjoint small subgraphs), we deduce that
\begin{equation}
\label{eqn:decC2}
h_s(Dec_{\lg n}C) = \Omega\left(\left(\frac47 \right)^{\log_7 s} \right),
\end{equation}
where $h_s$ is the edge expansion for sets of size at most $s$.

Let $S$ be the set of vertices of $Dec_{\lg n}C$ that correspond to computations performed by the given processor. Set $s=|S|=\Theta\lt(\frac{n^{\omega_0}}{P}\rt)$. By equation~\eqref{eqn:decC2}, the number of edges between $S$ and $\overline{S}$ is 
$$|E(S,\overline{S})|= \Omega \left( s \cdot h_s(Dec_{\lg n}C) \right)=\Omega \left(\frac{n^2}{P^{2/\omega_0}} \right),$$
and because $Dec_{\lg n}C$ is of bounded degree (Fact 9 there) and each vertex is computed by only one processor, the number of words moved is $\Theta(|E(S,\overline S)|)$ and the result follows.
\end{proof}

\subsection{Classical Matrix Multiplication}
\label{sec:classical}

In this section, we prove a memory-independent lower bound for classical matrix multiplication of $\Omega(n^2/P^{2/3})$ words.  The same result appears elsewhere in the literature, under slightly different assumptions: in the LPRAM model \cite{ACS90}, where no data exists in the (unbounded) local memories at the start of the algorithm; in the distributed-memory model \cite{IronyToledoTiskin04}, where the local memory size is assumed to be $M=\Theta(n^2/P^{2/3})$; and in the distributed-memory model \cite{SolomonikDemmel11}, where the algorithm is assumed to perform a certain amount of input replication. Our bound is for the distributed memory model, holds for any $M$, and assumes no specific communication pattern.

Recall the following special case of the Loomis-Whitney geometric bound:
\begin{lemma}
\label{lemma:LW}
\emph{\cite{LoomisWhitney49}} Let $V$ be a finite set of lattice points in ${\bf R}^3$, i.e., points $(x,y,z)$ with integer coordinates. Let $V_x$ be the projection of $V$ in the $x$-direction, i.e., all points $(y,z)$ such that there exists an $x$ so that $(x,y,z) \in V$. Define $V_y$ and $V_z$ similarly. Let $| \cdot |$ denote the cardinality of a set. Then $|V| \leq \sqrt{ |V_x| \cdot |V_y| \cdot |V_z| }$.
\end{lemma}

Using Lemma~\ref{lemma:LW} (in a similar way to \cite{BallardDemmelHoltzSchwartz11a,IronyToledoTiskin04}), we can describe the ratio between the number of scalar multiplications a processor performs and the amount of data it must access.

\begin{lemma}
\label{lemma:classical}
Suppose a processor has $I$ words of initial data at the start of an algorithm, performs $\Theta(n^3/P)$ scalar multiplications within classical matrix multiplication, and then stores $O$ words of output data at the end of the algorithm.  Then the processor must send or receive at least $\Omega(n^2/P^{2/3}) - I - O$ words during the execution of the algorithm.
\end{lemma}
\begin{proof}
We follow the proofs in \cite{BallardDemmelHoltzSchwartz11a,IronyToledoTiskin04}. Consider a discrete $n\times n\times n$ cube where the lattice points correspond to the scalar multiplications within the matrix multiplication $A\cdot B$ (\emph{i.e.}, lattice point $(i,j,k)$ corresponds to the scalar multiplication $a_{ik}\cdot b_{kj}$).  Then the three pairs of faces of the cube correspond to the two input and one output matrices.

The projections on the three faces correspond to the input/output elements the processor has to access (and must communicate if they are not in its local memory).  By Lemma~\ref{lemma:LW}, and the fact that $\sqrt{ |V_x| \cdot |V_y| \cdot |V_z| }\leq \sqrt{ \frac16 (|V_x| + |V_y| + |V_z|)^3}$, the number of words the processor must access is at least $\sqrt[3]{6} \; |V|^{2/3} = \Omega(n^2/P^{2/3})$.  Since the processor starts with $I$ words and ends with $O$ words, the result follows.
\end{proof}

\begin{theorem}
\label{thm:classical}
Suppose a parallel algorithm performing classical dense matrix multiplication begins with one copy of the input matrices and minimizes computational costs in an asymptotic sense.   Then, for sufficiently large $P$,\footnote{The theorem applies to any \(P\geq 2\) with a strict enough assumption on the load balance.} some processor must send or receive at least $\Omega\lt(\frac{n^2}{P^{2/3}}\rt)$.
\end{theorem}
\begin{proof}
At the end of the algorithm, every element of the output matrix must be fully computed and exist in some processor's local memory (though multiples copies of the element may exist in multiple memories).  For each output element, we designate one memory location as the output and disregard all other copies.  For each of the $n^2$ designated memory locations, we consider the $n$ scalar multiplications whose results were used to compute its value and disregard all other redundantly computed scalar multiplications.

In order to minimize computational costs asymptotically, the running time for classical dense matrix multiplication must be $O(n^3/P)$.  This is possible only if at least a constant fraction of the processors perform $\Theta\left(\frac{n^3}{P} \right)$ of the scalar multiplications corresponding to designated outputs.

Since there exists only one copy of the input matrices and designated output--$O(n^2)$ words of data--some processor which performs $\Theta(n^3/P)$ multiplications must start and end with no more than $I+O=O(n^2/P)$ words of data.  Thus, by Lemma~\ref{lemma:classical}, some processor must read or write $\Omega(n^2/P^{2/3})-O(n^2/P)=\Omega(n^2/P^{2/3})$ words of data.
\end{proof}

\section{Limits of Strong Scaling}
\label{sec:strongscaling}

In this section we present limits of strong scaling of matrix multiplication algorithms.  
These are immediate implications of the memory independent communication lower bounds proved in Section~\ref{sec:LB}.
Roughly speaking, the memory-dependent communication-cost lower-bound is of the form
$\Omega\lt( f(n,M)/P\rt)$ for both classical and Strassen matrix
multiplication algorithms. However, the memory independent lower bounds
are of the form $\Omega\lt(f(n,M)/P^c\rt)$ where $c<1$ 
(see Table~\ref{tbl:summary}).
This implies that strong scaling is not possible when the memory-independent bound dominates. We make this formal below.

\begin{table}[t]
\centering
\small
\begin{tabular}{|c||c||c|} \hline
 & {Classical} & {Strassen} \\ \hline
Memory-dependent &
\multirow{2}{*}{$\Omega \lt(\frac{n^3}{P\sqrt{M}} \rt)$} &
\multirow{2}{*}{$\Omega \lt(\frac{n^{\omega_0}}{PM^{\omega_0/2-1}}\rt)$} \\
lower bound & & \\
\hline
Memory-independent &
\multirow{2}{*}{$\Omega \lt( \frac{n^2}{P^{2/3}} \rt)$} &
\multirow{2}{*}{$\Omega \lt( \frac{n^2}{P^{2/\omega_0}} \rt)$} \\
lower bound & & \\
\hline
Perfect strong &
\multirow{2}{*}{$P = O\left(\frac{n^3}{M^{3/2}} \right)$} &
\multirow{2}{*}{$P = O\left(\frac{n^{\omega_0}}{M^{\omega_0/2}}\right)$} \\
scaling range & & \\
\hline\hline
Attaining algorithm &\cite{SolomonikDemmel11} &\cite{BallardDemmelHoltzLipshitzSchwartz12a} \\
\hline
\end{tabular}
\caption{Bandwidth-cost lower bounds for matrix multiplication and perfect strong scaling ranges.  The classical memory dependent bound is due to \cite{IronyToledoTiskin04}, and the Strassen memory dependent bound is due to \cite{BallardDemmelHoltzSchwartz11b}.  The memory-independent bounds are proved here, though variants of the classical bound appear in \cite{ACS90,IronyToledoTiskin04,SolomonikDemmel11}.}
\label{tbl:summary}
\end{table}

\begin{corollary}\label{cor:strong-strassen}
Suppose a parallel algorithm performing Strassen's matrix multiplication minimizes bandwidth and computational costs in an asymptotic sense and performs no redundant computation.  Then the algorithm can achieve perfect strong scaling only for $P = O\lt(\frac{n^{\omega_0}}{M^{\omega_0/2}}\rt)$.
\end{corollary}
\begin{proof}
By \cite{BallardDemmelHoltzSchwartz11b}, any parallel algorithm performing matrix multiplication based on Strassen moves at least $\Omega\lt(\frac{n^{\omega_0}}{PM^{\omega_0/2-1}}\rt)$ words.  By Theorem~\ref{thm:Strassen}, a parallel algorithm that minimizes computational costs and performs no redundant computation moves at least $\Omega\lt(\frac{n^2}{P^{2/\omega_0}}\rt)$ words.  This latter bound dominates in the case $P = \Omega\lt(\frac{n^{\omega_0}}{M^{\omega_0/2}}\rt)$.  Thus, while a communication-optimal algorithm will strongly scale perfectly up to this threshold, after the threshold the communication cost will scale as $1/P^{2/\omega_0}$ rather than $1/P$.
\end{proof}

\begin{corollary}\label{cor:strong-cubic}
Suppose a parallel algorithm performing classical dense matrix multiplication starts and ends with one copy of the data and minimizes bandwidth and computational costs in an asymptotic sense.  Then the algorithm can achieve perfect strong scaling only for $P = O\lt(\frac{n^3}{M^{3/2}}\rt)$.
\end{corollary}
\begin{proof}
By \cite{IronyToledoTiskin04}, any parallel algorithm performing matrix multiplication moves at least $\Omega\lt(\frac{n^3}{P\sqrt M}\rt)$ words.  By Theorem~\ref{thm:classical}, a parallel algorithm that starts and ends with one copy of the data and minimizes computational costs moves at least $\Omega\lt(\frac{n^2}{P^{2/3}}\rt)$ words.  This latter bound dominates in the case $P = \Omega\lt(\frac{n^3}{M^{3/2}}\rt)$.  Thus, while a communication-optimal algorithm will strongly scale perfectly up to this threshold, after the threshold the communication cost will scale as $1/P^{2/3}$ rather than $1/P$.
\end{proof}

\includefigs{
\begin{figure}[t]
\begin{center}
\scalebox{.6}{\includegraphics{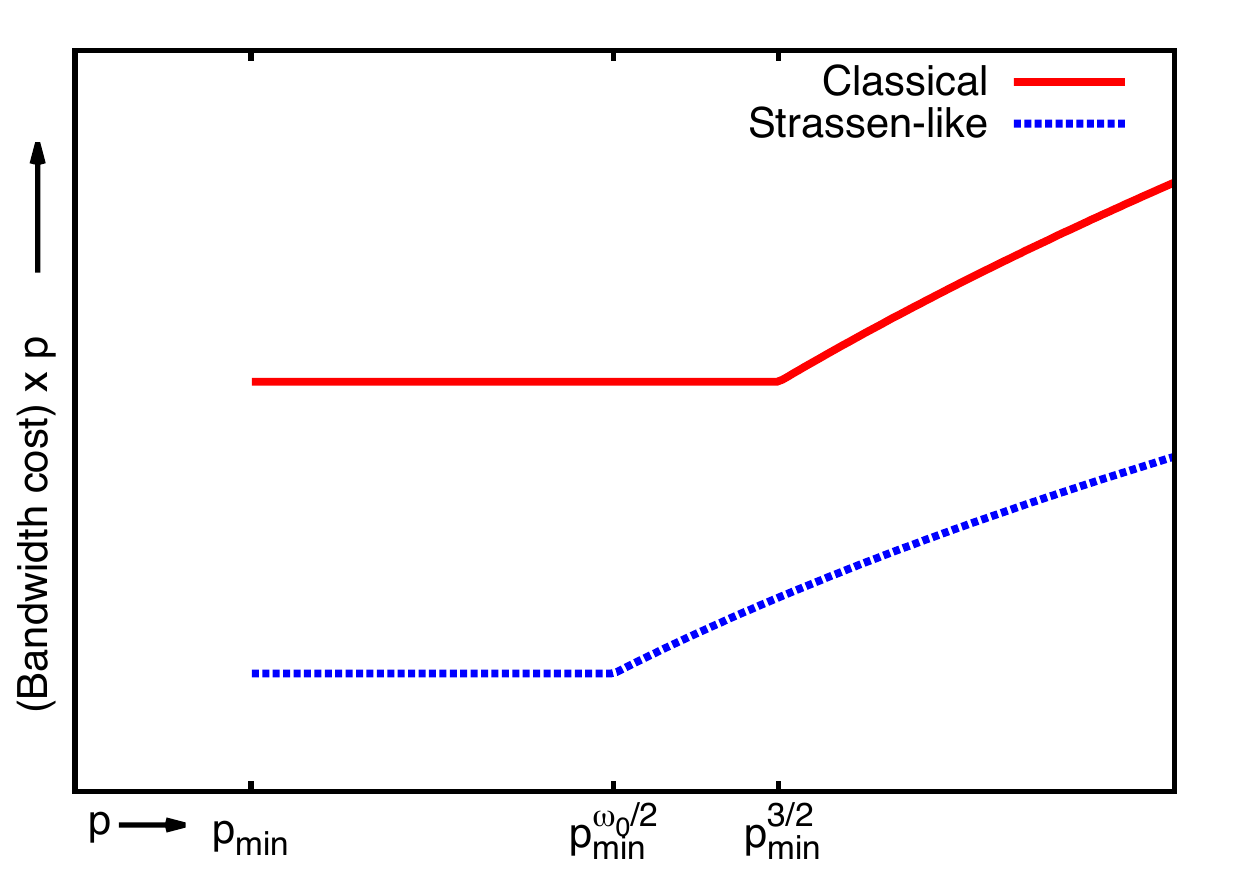}}
\protect\caption{Bandwidth costs and strong scaling of matrix multiplication: classical vs. Strassen-based.  Horizontal lines correspond to perfect strong scaling.  P$_\text{min}$ is the minimum number of processors required to store the input and output matrices.}
\label{fig:scaling}
\end{center}
\end{figure}
}

In Figure~\ref{fig:scaling} we present the asymptotic communication costs of classical and Strassen-based algorithms for a fixed problem size as the number of processors increases.  Both of the perfectly strong scaling algorithms stop scaling perfectly above some number of processors, which depends on the matrix size and the available local memory size.

Let $P_\text{min}=\Theta\lt(\frac{n^2}{M}\rt)$ be the minimum number of processors required to store the input and output matrices.  By Corollaries \ref{cor:strong-strassen} and \ref{cor:strong-cubic} the perfect strong scaling range is $P_\text{min} \leq P \leq P_\text{max}$ where $P_\text{max}=\Theta(P_\text{min}^{3/2})$ in the classical case and $P_\text{max}=\Theta(P_\text{min}^{\omega_0/2})$ in the Strassen case.

Note that the perfect strong scaling range is larger for the classical case, though the communication costs are higher.  

\section{Extensions and Open Problems}
\label{sec:discussion}

The memory-independent bound and perfect strong scaling bound of Strassen's matrix multiplication (Theorem~\ref{thm:Strassen} and Corollary~\ref{cor:strong-strassen}) apply to other Strassen-like algorithms, as defined in \cite{BallardDemmelHoltzSchwartz11a}, with $\omega_0$ being the exponent of the total arithmetic count, provided that $Dec_{\lg n} C$ is connected.   The proof follows that of Theorem \ref{thm:Strassen} and of Corollary \ref{cor:strong-strassen}, but uses Claim 18 of \cite{BallardDemmelHoltzSchwartz11b} instead of Fact 9 there, and replaces Lemma 10 there with its extension.

The memory-dependent bound of classical matrix multiplication of \cite{IronyToledoTiskin04} was generalized in \cite{BallardDemmelHoltzSchwartz11a} to algorithms which perform computations of the form
\begin{equation}
\label{eqn:3-nested-loops}
\text{Mem}(c(i,j)) = f_{ij} ( g_{ijk} (\text{Mem}(a(i,k)),\text{Mem}(b(k,j)))),
\end{equation}
where $\text{Mem}(i)$ denotes the argument in memory location $i$ and $f_{ij}$ and $g_{ijk}$ are functions which depend non-trivially on their arguments (see \cite{BallardDemmelHoltzSchwartz11a} for more detailed definitions).  

The memory-independent bound of classical matrix multiplication (Theorem~\ref{thm:classical}) applies to these other algorithms as well.  If the algorithm begins with one copy of the input data and minimizes computational costs in an asymptotic sense, then, for sufficiently large $P$, some processor must send or receive at least $\Omega \lt( \lt( \frac{G}{P}\rt)^{2/3} -\frac{D}{P} \rt)$ words, where $G$ is the total number of $g_{ijk}$ computations and $D$ is the number of non-zeros in the input and output.  The proof follows that of Lemma \ref{lemma:classical} and Theorem \ref{thm:classical}, setting $|V| = G$ (instead of $n^3$), replacing $n^3/P$ with $G/P$ , and setting $I+O = O(D/P)$ (instead of $O(n^2/P)$).

Algorithms which fit the form of equation \eqref{eqn:3-nested-loops} include LU and Cholesky decompositions, sparse matrix-matrix multiplication, as well as algorithms for solving the all-pairs-shortest-paths problem.  Only a few of these have parallel algorithms which attain the lower bounds in all cases.  In several cases, it seems likely that one can prove better bounds than those presented here, thus obtaining a stricter bound on perfect strong scaling.

We also believe that our bounds can be generalized to QR decomposition and other orthogonal transformations, fast linear algebra, fast Fourier transform, and other recursive algorithms.



\bibliographystyle{acm} 
\bibliography{brief}

\begin{thebibliography}{1}

\bibitem{ACS90}
{\sc Aggarwal, A., Chandra, A.~K., and Snir, M.}
\newblock Communication complexity of {PRAM}s.
\newblock {\em Theoretical Computer Science 71}, 1 (1990), 3 -- 28.

\bibitem{BallardDemmelHoltzLipshitzSchwartz12a}
{\sc Ballard, G., Demmel, J., Holtz, O., Lipshitz, B., and Schwartz, O.}
\newblock Communication-optimal parallel algorithm for {Strassen's} matrix
  multiplication, 2012.
\newblock Submitted to SPAA.

\bibitem{BallardDemmelHoltzSchwartz11b}
{\sc Ballard, G., Demmel, J., Holtz, O., and Schwartz, O.}
\newblock Graph expansion and communication costs of fast matrix
  multiplication.
\newblock In {\em SPAA '11: Proceedings of the 23rd Annual Symposium on
  Parallelism in Algorithms and Architectures\/} (New York, NY, USA, 2011),
  ACM, pp.~1--12.

\bibitem{BallardDemmelHoltzSchwartz11a}
{\sc Ballard, G., Demmel, J., Holtz, O., and Schwartz, O.}
\newblock Minimizing communication in numerical linear algebra.
\newblock {\em SIAM J. Matrix Analysis Applications 32}, 3 (2011), 866--901.

\bibitem{IronyToledoTiskin04}
{\sc Irony, D., Toledo, S., and Tiskin, A.}
\newblock Communication lower bounds for distributed-memory matrix
  multiplication.
\newblock {\em J. Parallel Distrib. Comput. 64}, 9 (2004), 1017--1026.

\bibitem{LoomisWhitney49}
{\sc Loomis, L.~H., and Whitney, H.}
\newblock An inequality related to the isoperimetric inequality.
\newblock {\em Bulletin of the {AMS} 55\/} (1949), 961--962.

\bibitem{SolomonikDemmel11}
{\sc Solomonik, E., and Demmel, J.}
\newblock {Communication-optimal parallel 2.5D matrix multiplication and LU
  factorization algorithms}.
\newblock In {\em Euro-Par '11: Proceedings of the 17th International European
  Conference on Parallel and Distributed Computing\/} (2011), Springer.

\end{thebibliography}

\end{document}